\documentclass[conference]{IEEEtran}
\usepackage{epsfig}
\usepackage{times}
\usepackage{float}
\usepackage{afterpage}
\usepackage{amsmath}
\usepackage{amstext}
\usepackage{amssymb,bm}
\usepackage{latexsym}
\usepackage{color}
\usepackage{graphicx}
\usepackage{amsmath}
\usepackage{amsthm}
\usepackage{graphicx}
\usepackage{subfigure}
\usepackage{booktabs}
\usepackage{multicol}
\usepackage{lipsum}
\usepackage{dblfloatfix}
\usepackage{mathrsfs}
\usepackage{cite}
\usepackage{tikz}
\usepackage{pgfplots}
\pgfplotsset{compat=newest}
 \usepackage{enumitem}
\usepackage{makecell}

\allowdisplaybreaks

\newcommand{\rmd}{\mathrm{d}}
\newcommand{\bbE}{\mathbb{E}}\newcommand{\rme}{\mathrm{e}}

\newcommand{\bbP}{\mathbb{P}}

\newcommand{\bbR}{\mathbb{R}}

\newcommand{\sfC}{\mathsf{C}}
\newcommand{\bfD}{\mathbf{D}}\newcommand{\sfD}{\mathsf{D}}

\newcommand{\bfH}{\mathbf{H}}\newcommand{\sfH}{\mathsf{H}}
\newcommand{\sfI}{\mathsf{I}}

\newcommand{\bfQ}{\mathbf{Q}}

\newcommand{\bft}{\mathbf{t}}\newcommand{\sfT}{\mathsf{T}}
\newcommand{\bfU}{\mathbf{U}}
\newcommand{\bfV}{\mathbf{V}}

\newcommand{\bfX}{\mathbf{X}}\newcommand{\bfx}{\mathbf{x}}
\newcommand{\bfY}{\mathbf{Y}}\newcommand{\bfy}{\mathbf{y}}
\newcommand{\bfZ}{\mathbf{Z}}\newcommand{\bfz}{\mathbf{z}}

\newcommand{\cA}{\mathcal{A}}

\newcommand{\cE}{\mathcal{E}}

\newtheoremstyle{mystyle}
{}
{}
{\itshape}
{}
{\bfseries}
{}
{.5em}
{}

\newtheoremstyle{remark}
{}
{}
{}
{}
{\itshape}
{}
{.5em}
{}

\makeatletter
\def\thmhead@plain#1#2#3{%
  \thmname{#1}\thmnumber{\@ifnotempty{#1}{ }\@upn{#2.}}%
  \thmnote{ \textsf{\the\thm@notefont\textit{#3}.}}}
\let\thmhead\thmhead@plain
\makeatother

\theoremstyle{mystyle}
\newtheorem{theorem}{Theorem}
\theoremstyle{mystyle}
\newtheorem{lemma}{Lemma}
\theoremstyle{mystyle}
\newtheorem{prop}{Proposition}
\theoremstyle{mystyle}
\theoremstyle{mystyle}
\newtheorem{definition}{Definition}
\theoremstyle{remark}
\newtheorem{rem}{Remark}
\theoremstyle{mystyle}
\theoremstyle{mystyle}
\theoremstyle{mystyle}
\theoremstyle{discussion}
\theoremstyle{mystyle}
\theoremstyle{mystyle}
\newcommand\independent{\protect\mathpalette{\protect\independent}{\perp}}
\def\independent#1#2{\mathrel{\rlap{$#1#2$}\mkern2mu{#1#2}}}

\def\squarebox#1{\hbox to #1{\hfill\vbox to #1{\vfill}}}

\newcommand{\del}{\partial}

\newcommand{\kl}[2]{{\sfD}\left(\left.#1 \, \right\| #2 \right)}

\makeatletter
\newcommand{\vast}{\bBigg@{4}}
\newcommand{\Vast}{\bBigg@{5}}
\newcommand{\Gigantic}{\bBigg@{8}}

\makeatother
\allowdisplaybreaks

\newcommand{\supp}{{\mathsf{supp}}}

\title{On $2 \times 2$ MIMO Gaussian Channels with a Small Discrete-Time Peak-Power Constraint} 

\author{
\IEEEauthorblockN{Alex Dytso$^{*}$, Luca Barletta$^{\dagger}$,  and Gerhard Kramer$^{**}$}
$^{*}$ Qualcomm Flarion Technologies, Bridgewater,  NJ 08807, USA.
Email: odytso2@gmail.com  \\
$^{\dagger}$ Politecnico di Milano, Milano, 20133, Italy. Email:  luca.barletta@polimi.it \\
$^{**}$ Technical University of Munich, Munich, 80333, Germany.
Email: gerhard.kramer@tum.de
}

\begin{document}

\maketitle

\begin{abstract}
A multi-input multi-output (MIMO) Gaussian channel with two transmit antennas and two receive antennas is studied that is subject to an input peak-power constraint. The capacity and the capacity-achieving input distribution are unknown in general. The problem is shown to be equivalent to a channel with an identity matrix but where the input lies inside and on an ellipse with principal axis length $r_p$ and minor axis length $r_m$. If $r_p \le \sqrt{2}$, then the capacity-achieving input has support on the ellipse. A sufficient condition is derived under which a two-point distribution is optimal. Finally, if $r_m < r_p \le \sqrt{2}$, then the capacity-achieving distribution is discrete.
\end{abstract}

\section{Introduction} 
The capacity of multi-input multi-output (MIMO) Gaussian noise channels with an average-power constraint is known.  However, in practice one is also interested in a peak-power constraint, and determining the corresponding capacity is challenging. The goal of this work is to advance characterizing the capacity under a peak-power constraint $\|\bfX\|\le 1$, namely
\begin{equation}
    \max_{ P_{\bfX} :  \| \bfX\| \le 1} I(\bfX; \mathsf{\bfH} \bfX+\bfZ) \label{eq:Original_cap}
\end{equation}
where $\mathsf{\bfH}\in \bbR^{2 \times 2}$ is a channel matrix known to both the receiver and the transmitter, $\bfZ \in \bbR^n $ is a standard normal vector independent of the input $\bfX \in \bbR^n$, and $\| \cdot \|$ is the Euclidean norm.

\subsection{Problem Reformulation}
We reformulate the problem as follows:
\begin{equation}
\sfC(r_p,r_m)=\max_{ P_{\bfX}: \bfX \in \cE(r_p,r_m)} I (\bfX; \bfX +\bfZ) , \,  r_p,r_m \ge 0, \label{eq:Capacity_Def}
\end{equation} 
where $\cE(r_p,r_m)$ is the set of points inside and on the two-dimensional ellipse:
\begin{equation}
\cE(r_p,r_m) = \left \{ (x_1,x_2) \in \bbR^2: \frac{x_1^2}{r_p^2} +\frac{x_2^2}{r_m^2}  \le 1 \right \} ,
\end{equation}
where we assume $r_p \ge r_m$ and refer to $r_p$ as a primary radius and $r_m$ as a minor radius. The equivalence of \eqref{eq:Original_cap} and \eqref{eq:Capacity_Def} is shown in Appendix~\ref{sec:proof_of_equivalence}. 

\label{sec:some_mysteries}
To motivate our discussion, suppose $r_m=0$, i.e., the channel is one-dimensional. The capacity-achieving input distribution is discrete with finite support~\cite{smith1971information}.  Moreover, if $r_p \le \bar{r}_p \approx 1.6$ then the optimal input distribution is uniform on $[-r_p,r_p]$ \cite{sharma2010transition}. Suppose next that $r_p=r_m$. The optimal input distribution is uniform on a circle of radius $r_p$ for all $r_p \le \tilde{r}_p \approx 2.4 $ \cite{shamai1995capacity,dytsoEstPerspective}.
One goal of this work is to explore what happens when $0 <r_m <r_p\le \tilde{r}_p $.

\subsection{Literature Review}
The single-input single-output (SISO) version of the channel was considered in \cite{smith1971information} and the capacity-achieving input distribution is discrete with finite support. The optimality of the two-point distribution was characterized in \cite{sharma2010transition}. For further insights into the structure of the capacity-achieving distribution, the reader is referred to \cite{dytso2019capacity}. Upper and lower bounds on the capacity appear in \cite{capBoundSISO_amp,rassouli2016-2,mckellips2004simple,dytso2017generalized}.

The MIMO problem seems considerably more difficult than the SISO problem. For example, the support of the capacity-achieving distribution is, in general, not a union of isolated points but a union of lower-dimensional manifolds \cite{e21020200,eisen2023capacity,chan2005capacity}. Little is known about the structure of these manifolds. The structure is known when the channel matrix is proportional to an identity matrix, in which case the support consists of concentric shells \cite{dytsoEstPerspective,shamai1995capacity,rassouli2016capacity, favano2021capacity}. Upper and lower bounds on the capacity of general MIMO channels were considered in \cite{capDuman,genMIMOcapBoundsGlob, favano2022sphere}.

\subsection{Outline and Contributions}

This paper is organized as follows. Sec.~\ref{subsec:notation} describes the notation.  Sec.~\ref{sec:Main_results} presents our main results.  Specifically, Sec.~\ref{sec:concentration_on_boundary} shows that for a sufficiently small ellipse the capacity-achieving distribution is supported on the boundary of the input space. Sec.~\ref{sec:symmetry} establishes symmetry properties and Sec.~\ref{sec:two_point} proves a sufficient condition for optimality of a two-point distribution. Sec.~\ref{sec:on_discretness} shows that, if the ellipse is small enough, then the capacity-achieving distribution is discrete.  Sec.~\ref{sec:proofs} provides proofs.  Finally, Sec.~\ref{sec:disc_conc} concludes with a discussion and future directions. 

\subsection{Notation}
\label{subsec:notation}
The probability of a measurable subset $\cA \subseteq \bbR^n$ is written as $P_{\bfX}(\cA) = \bbP[ \bfX~\in~\cA]$. The support of $P_\bfX$ is
\begin{align}
\supp(P_{\bfX})=\{ \bfx:&  \text{ $P_{\bfX}( \mathcal{D})>0$ for every open set $ \mathcal{D} \ni \bfx $ } \}. \nonumber
\end{align} 
If $\bfX$ is discrete, we write $P_\bfX(\bfx)$ 
for its probability mass function (pmf). The probability density function (pdf) of a standard Gaussian is denoted by $\phi$. The relative entropy of $P$ and $Q$ is $\kl{P}{Q}$.   For  two random vectors $\bfX$ and $\bfY $,  $\bfX \stackrel{d}{=} \bfY$ denotes equality in distribution. 






\section{Main Results} 
\label{sec:Main_results}

\subsection{Karush-Kuhn-Tucker (KKT) Conditions }

\begin{lemma} \label{lem:KKT_conditions} The optimal $P_{\bfX^\star}$ satisfies
\begin{align}
\kl{ P_{\bfY|\bfX}(\cdot | \bfx) }{  P_{\bfY^\star} } &\le \sfC(r_p,r_m), \quad \bfx \in \cE(r_p,r_m) \label{eq:original_KKT_inequality}\\
\kl{ P_{\bfY|\bfX}(\cdot | \bfx) }{  P_{\bfY^\star} }& =  \sfC(r_p,r_m), \quad \bfx \in {\rm supp}(P_{\bfX^\star} )
\end{align}
\end{lemma} 
\begin{proof}
    This is an extension of the case $r_p=r_m$ treated in~\cite{shamai1995capacity}.
\end{proof}

Note that the KKT conditions imply that the set ${\rm supp}(P_{\bfX^\star} )$ is the set of global maxima of the function $\bfx \mapsto  \kl{ P_{\bfY|\bfX}(\cdot | \bfx) }{  P_{\bfY^\star} }$. This fact will be used several times. 

\subsection{When is the Support Concentrated on the Ellipse?  }
\label{sec:concentration_on_boundary}
The theorem below identifies a sufficient condition so the support of the optimal input distribution lies on the ellipse. To state the theorem, we need the following definition and lemma.
 
\begin{definition} Let $f$ be a  real-valued function that is twice continuously differentiable on an open set $\mathcal{G} \subset \mathbb{R}^n$. Then $f$ is \emph{subharmonic} if $\nabla^2 f \ge 0$ on $\mathcal{G}$, where $\nabla^2$ is the Laplacian, i.e., the trace of the Hessian matrix. 
\end{definition} 

\begin{lemma}\label{lem:sub_harmonic_KL}
The function $\bfx \mapsto \kl{ P_{\bfY|\bfX}(\cdot | \bfx) }{  P_{\bfY} }$ is subharmonic provided that  $ \| \bfX\|  \le \sqrt{2}$.
\end{lemma} 
\begin{proof}
The Hatsel-Nolte identity \cite{hatsell1971some, dytso2022conditional} gives
\begin{equation}
\mathsf{H}_{\bfy} \log f_{\bfY}(\bfy) =  \mathsf{Var} (\bfX| \bfY=\bfy) - \sfI,
\end{equation} 
where $\mathsf{H}_{\bfy} $ is the Hessian and $\sfI$ is an identity matrix. We further have
\begin{align}
&\mathsf{H}_{\bfx} \kl{ P_{\bfY|\bfX}(\cdot | \bfx) }{ P_{\bfY^\star} } = - \mathsf{H}_{\bfx}  \bbE \left[ \log f_{\bfY}(\bfx+\bfZ) \right] \nonumber \\
&=  \sfI -   \bbE \left[  \mathsf{Var} (\bfX| \bfY=\bfx+\bfZ) \right] . \label{eq:Hessian_of_KL}
\end{align} 
Taking the trace of \eqref{eq:Hessian_of_KL}, we arrive at
\begin{align*}
 \nabla^2_{\bfx} \kl{ P_{\bfY|\bfX}(\cdot | \bfx) }{  P_{\bfY^\star} }&= 2 -   \bbE \left[  \rm{Tr} ( \mathsf{Var} (\bfX| \bfY=\bfx+\bfZ)) \right] \\
 &\ge  2 -  2 =0
\end{align*} 
where we used $\| \bfX\| \le \sqrt{2}$.
\end{proof}

\begin{theorem} \label{thm:boundary_condition} For $r_m \le r_p \le \sqrt{2}$, we have
\begin{equation}
{\rm supp}(P_{\bfX^\star} ) \subseteq  \partial \cE(r_p,r_m), 
\end{equation} 
where $\partial \cE(r_p,r_m)$ is the boundary (an ellipse) of $\cE(r_p,r_m)$, 
\end{theorem}

\begin{proof}
Lemma~\ref{lem:sub_harmonic_KL} shows that $\kl{ P_{\bfY|\bfX}(\cdot | \bfx) }{  P_{\bfY} } $ is sub-harmonic if the support of $P_{\bfX}$ is in a sufficiently small ball. Next, note that non-constant subharmonic functions are maximized on the boundary of a closed non empty domain \cite{stein2010complex}. Since $\bfx \mapsto \kl{ P_{\bfY|\bfX}(\cdot | \bfx) }{  P_{\bfY} } $ is subharmonic  for $r_p\le \sqrt{2}$, the maximum of $\kl{ P_{\bfY|\bfX}(\cdot | \bfx) }{  P_{\bfY} } $ over $\cE(r_p,r_m)$ occurs on the boundary of  $\cE(r_p,r_m)$ for $r_m \le r_p\le \sqrt{2}$. 
\end{proof} 

\begin{rem} 
Lemma~\ref{lem:sub_harmonic_KL} generalizes to vector channels of  any dimension $n$ if $ \| \bfX\|  \le \sqrt{2}$ is replaced with $ \| \bfX\|  \le \sqrt{n}$.
\end{rem}

We now focus on $r_m \le r_p \le \sqrt{2}$, which we call the \emph{small peak-power regime}. Theorem~\ref{thm:boundary_condition} simplifies the KKT conditions.

\begin{lemma} \label{lem:new_KKT_conditions} $P_{\bfX^\star}$ is optimal if and only if
\begin{align}
\kl{ P_{\bfY|\bfX}(\cdot | \bfx) }{  P_{\bfY^\star} } &\le \sfC(r_p,r_m) \label{eq:new_kkt_inequality}, \quad \bfx \in \partial \cE(r_p,r_m)\\
\kl{ P_{\bfY|\bfX}(\cdot | \bfx) }{  P_{\bfY^\star} }& =  \sfC(r_p,r_m), \quad \bfx \in {\rm supp}(P_{\bfX^\star} ) \label{eq:new_KKT_equality}
\end{align}
\end{lemma} 
Note that in \eqref{eq:new_kkt_inequality}, we check only on the ellipse rather than on the region enclosed by the ellipse.

\subsection{Symmetry}
\label{sec:symmetry}
We prove the following symmetry result. 
\begin{prop}  \label{prop:symmetry}
The marginals $P_{X_1^\star}$ and $P_{X_2^\star}$ are symmetric. Also, if $0<r_m \le r_p \le \sqrt{2}$, then for  $\bfx \in \cE(r_p,r_m)$, we have
\begin{align}
&P_{X_2^\star|X_1^\star=x_1}(x_2) \notag\\
&=   \frac{1}{2} \delta \left(x_2 -r_m \sqrt{1-\frac{x_1^2}{r_p^2}}  \right)+ \frac{1}{2} \delta \left(x_2 +r_m \sqrt{1-\frac{x_1^2}{r_p^2}}  \right).
\label{eq:second-property}
\end{align}
\end{prop}
\begin{IEEEproof}
Observe that 
\begin{align}
\sfC(r_p,r_m) &=I\left(\bfX^\star; \bfX^\star + \bfZ \right) \nonumber \\
&= I \left((-1) \bfX^\star; (-1)( \bfX^\star + \bfZ) \right) \nonumber \\
&= I \left( - \bfX^\star; - \bfX^\star +\bfZ \right) \label{eq:symmetry of Z}
\end{align}
where \eqref{eq:symmetry of Z} follows because $-\bfZ \stackrel{d}{=}\bfZ$.
Note that $-\bfX^\star \in \cE(r_p,r_m)$, and \eqref{eq:symmetry of Z} implies the input $-\bfX^\star$ is also capacity-achieving. However, since the capacity-achieving distribution is unique  \cite[Thm.~1]{chan2005capacity}, we have $-\bfX^\star \stackrel{d}{=} \bfX^\star$. Next, marginalizing with respect to $X_2^\star$ gives $X_1^\star \stackrel{d}{=} -X_1^\star$, i.e., $X_1^\star$ is symmetric. Similarly, the marginal of $X_2^*$ is symmetric. 

To show~\eqref{eq:second-property}, we have $\bfX^\star \in \partial \cE(r_p,r_m)$ for $r_m \le r_p \le \sqrt{2}$. Given $X_1^\star = x_1$, this implies
\begin{align}
   & P_{X_2^\star|X_1^\star=x_1}(x_2)  \notag\\
    &=  p\, \delta \left(x_2 -r_m \sqrt{1-\frac{x_1^2}{r_p^2}}  \right)+ q \,  \delta \left(x_2 +r_m \sqrt{1-\frac{x_1^2}{r_p^2}}  \right)
\end{align}
for some $0 \le p,q \le 1$ such that $p + q =1$. Finally, $p=q =\frac{1}{2}$ follows by the symmetry of $X_2^\star$.
\end{IEEEproof}

\subsection{When are Two Mass  Points Optimal?} 
\label{sec:two_point}
We characterize the regime where two points are optimal.

\begin{theorem}\label{thm:optimal_two_points} If  $r_p \le  \sqrt{ 2} $ then 
\begin{equation}
P_{\bfX^\star}( -r_p,0  )=  P_{\bfX^\star}( r_p,0  )=\frac{1}{2}
\end{equation} 
is optimal if and only if $\frac{r_m^2}{2}$ is less than
\begin{equation}
\int_{-\infty}^\infty \left(   \phi(y-r_p) - \phi(y) \right) \log \frac{1}{  \phi(y-r_p) + \phi(y+r_p)}   {\rm d} y.  \label{eq:Condition_for_optimality_two_point}
\end{equation} 
\end{theorem} 
\begin{proof}
See Section~\ref{proof:thm:optimal_two_points}. 
\end{proof} 

Fig.~\ref{fig:Two-Four-Optimality} plots the boundary (in blue) of the region for which two mass points are optimal.

\begin{figure}
\center
%
%
\begin{tikzpicture}

\begin{axis}[%
width=6cm,
height=5cm,
at={(1.011in,0.642in)},
scale only axis,
xmin=0,
xmax=1.20112240878645,
xlabel style={font=\color{white!15!black}},
xlabel={$r_p$},
ymin=0,
ymax=1.20112240878645,
ylabel style={font=\color{white!15!black}},
ylabel={ $r_m$},
axis background/.style={fill=white},
axis x line*=bottom,
axis y line*=left,
xmajorgrids,
ymajorgrids,
legend style={legend cell align=left, align=left, draw=white!15!black}
]
\addplot [color=blue]
  table[row sep=crcr]{%
0.001	0.000999999499999258\\
0.025002448175729	0.0249946370397963\\
0.049004896351458	0.0489461597050922\\
0.073007344527187	0.0728135471640756\\
0.097009792702916	0.0965564824199653\\
0.121012240878645	0.1201356567899\\
0.145014689054374	0.143513032979662\\
0.169017137230103	0.166652059238893\\
0.193019585405832	0.189517831639667\\
0.217022033581561	0.212077205120733\\
0.24102448175729	0.234298856743487\\
0.265026929933019	0.256153306493062\\
0.289029378108748	0.277612901977613\\
0.313031826284477	0.298651773680113\\
0.337034274460206	0.319245767188818\\
0.361036722635935	0.339372358259138\\
0.385039170811664	0.359010555797007\\
0.409041618987393	0.378140797020099\\
0.433044067163122	0.39674483822998\\
0.457046515338851	0.414805643865923\\
0.48104896351458	0.432307275836116\\
0.505051411690309	0.449234784544232\\
0.529053859866038	0.465574102547482\\
0.553056308041767	0.481311941388649\\
0.577058756217496	0.496435691828183\\
0.601061204393225	0.510933327450822\\
0.625063652568954	0.524793311422338\\
0.649066100744683	0.538004506014182\\
0.673068548920412	0.550556084386893\\
0.697070997096141	0.562437444017668\\
0.72107344527187	0.573638121065408\\
0.745075893447599	0.584147704880432\\
0.769078341623328	0.593955751779174\\
0.793080789799057	0.603051697109984\\
0.817083237974786	0.611424764528303\\
0.841085686150515	0.619063871270877\\
0.865088134326244	0.625957528061864\\
0.889090582501973	0.632093732088758\\
0.913093030677702	0.637459851243949\\
0.937095478853431	0.6420424975208\\
0.96109792702916	0.645827387067813\\
0.985100375204889	0.648799183913459\\
1.00910282338062	0.650941323748465\\
1.03310527155635	0.652235813348936\\
1.05710771973208	0.652663000185682\\
1.0811101679078	0.652201305413856\\
1.10511261608353	0.650826911663126\\
1.12911506425926	0.648513394697735\\
1.15311751243499	0.645231284867787\\
1.17711996061072	0.640947540008651\\
1.20112240878645	0.63562490559503\\
};
\addlegendentry{Two-Point Boundary}

\addplot [color=red]
  table[row sep=crcr]{%
0.001	0.001\\
0.025002448175729	0.025002448175729\\
0.049004896351458	0.049004896351458\\
0.073007344527187	0.073007344527187\\
0.097009792702916	0.097009792702916\\
0.121012240878645	0.121012240878645\\
0.145014689054374	0.145014689054374\\
0.169017137230103	0.169017137230103\\
0.193019585405832	0.193019585405832\\
0.217022033581561	0.217022033581561\\
0.24102448175729	0.24102448175729\\
0.265026929933019	0.265026929933019\\
0.289029378108748	0.289029378108748\\
0.313031826284477	0.313031826284477\\
0.337034274460206	0.337034274460206\\
0.361036722635935	0.361036722635935\\
0.385039170811664	0.385039170811664\\
0.409041618987393	0.409041618987393\\
0.433044067163122	0.433044067163122\\
0.457046515338851	0.457046515338851\\
0.48104896351458	0.48104896351458\\
0.505051411690309	0.505051411690309\\
0.529053859866038	0.529053859866038\\
0.553056308041767	0.553056308041767\\
0.577058756217496	0.577058756217496\\
0.601061204393225	0.601061204393225\\
0.625063652568954	0.625063652568954\\
0.649066100744683	0.649066100744683\\
0.673068548920412	0.673068548920412\\
0.697070997096141	0.697070997096141\\
0.72107344527187	0.72107344527187\\
0.745075893447599	0.745075893447599\\
0.769078341623328	0.769078341623328\\
0.793080789799057	0.793080789799057\\
0.817083237974786	0.817083237974786\\
0.841085686150515	0.841085686150515\\
0.865088134326244	0.865088134326244\\
0.889090582501973	0.889090582501973\\
0.913093030677702	0.913093030677702\\
0.937095478853431	0.937095478853431\\
0.96109792702916	0.96109792702916\\
0.985100375204889	0.985100375204889\\
1.00910282338062	1.00910282338062\\
1.03310527155635	1.03310527155635\\
1.05710771973208	1.05710771973208\\
1.0811101679078	1.0811101679078\\
1.10511261608353	1.10511261608353\\
1.12911506425926	1.12911506425926\\
1.15311751243499	1.15311751243499\\
1.17711996061072	1.17711996061072\\
1.20112240878645	1.20112240878645\\
};
\addlegendentry{ $r_p=r_m$ (cont. uniform is optimal) }

\addplot [color=magenta, dashed]
  table[row sep=crcr]{%
0.01	0.01\\
0.024183975377506	0.024183975377506\\
0.038367950755012	0.038367950755012\\
0.052551926132518	0.052551926132518\\
0.0667359015100241	0.0666034996212567\\
0.0809198768875301	0.0807593344464105\\
0.0951038522650361	0.0947264862780926\\
0.109287827642542	0.108854180550894\\
0.123471803020048	0.12273691072552\\
0.137655778397554	0.136563359795441\\
0.15183975377506	0.150333527760658\\
0.166023729152566	0.164047414621171\\
0.180207704530072	0.177705020376979\\
0.194391679907578	0.190920678168146\\
0.208575655285084	0.204437581162194\\
0.22275963066259	0.217456255086896\\
0.236943606040096	0.230832455319182\\
0.251127581417602	0.243654145377419\\
0.265311556795108	0.256889642847945\\
0.279495532172614	0.269514349039716\\
0.29367950755012	0.282026493022079\\
0.307863482927626	0.294426074795034\\
0.322047458305132	0.306713094358579\\
0.336231433682638	0.318887551712716\\
0.350415409060144	0.330949446857444\\
0.36459938443765	0.342898779792764\\
0.378783359815156	0.354735550518675\\
0.392967335192662	0.366459759035177\\
0.407151310570168	0.377263630197052\\
0.421335285947674	0.388734573742384\\
0.43551926132518	0.400092955078308\\
0.449703236702686	0.410446577402548\\
0.463887212080192	0.421551693767303\\
0.478071187457698	0.431595770015668\\
0.492255162835204	0.442447621409253\\
0.50643913821271	0.452182151581745\\
0.520623113590216	0.462780738004161\\
0.534807088967722	0.472205722100779\\
0.548991064345228	0.481461862883284\\
0.563175039722734	0.49166648157277\\
0.57735901510024	0.500613076279401\\
0.591542990477746	0.509390827671919\\
0.605726965855252	0.519201478628474\\
0.619910941232758	0.527669683945119\\
0.634094916610265	0.53596904594765\\
0.648278891987771	0.544099564636069\\
0.662462867365277	0.553375545097932\\
0.676646842742783	0.561196517710477\\
0.690830818120289	0.568848647008909\\
0.705014793497795	0.576331932993227\\
0.719198768875301	0.5850732429604\\
0.733382744252807	0.592246982868845\\
0.747566719630313	0.599251879463176\\
0.761750695007819	0.606087932743395\\
0.775934670385325	0.614294572215877\\
0.790118645762831	0.620821079420222\\
0.804302621140337	0.627178743310454\\
0.818486596517843	0.634991415050005\\
0.832670571895349	0.641039532864363\\
0.846854547272855	0.646918807364608\\
0.861038522650361	0.65262923855074\\
0.875222498027867	0.6599072397956\\
0.889406473405373	0.665308124905859\\
0.903590448782879	0.670540166702004\\
0.917774424160385	0.677424200213934\\
0.931958399537891	0.682346695934205\\
0.946142374915397	0.688977464474965\\
0.960326350292903	0.693590414119363\\
0.974510325670409	0.698034520449648\\
0.988694301047915	0.704271321257477\\
1.00287827642542	0.708405881511888\\
1.01706225180293	0.714389417348549\\
1.03124622718043	0.718214431527086\\
1.04543020255794	0.723944702392577\\
1.05961417793544	0.72746017049524\\
1.07379815331295	0.732937176389561\\
1.08798212869046	0.738301620074474\\
1.10216610406796	0.741366839339502\\
1.11635007944547	0.746478018053246\\
1.13053405482297	0.7492336912424\\
1.14471803020048	0.754091604984974\\
1.15890200557799	0.758836956518139\\
1.17308598095549	0.761142380869659\\
1.187269956333	0.765634467431655\\
1.2014539317105	0.770013991784242\\
1.21563790708801	0.77428095392742\\
1.22982188246552	0.775995426679544\\
1.24400585784302	0.780009123851553\\
1.25818983322053	0.783910258814153\\
1.27237380859803	0.787698831567345\\
1.28655778397554	0.791374842111127\\
1.30074175935305	0.794938290445502\\
1.31492573473055	0.798389176570467\\
1.32910971010806	0.799090589437913\\
1.34329368548557	0.802288210591709\\
1.35747766086307	0.805373269536097\\
1.37166163624058	0.808345766271075\\
1.38584561161808	0.811205700796645\\
1.40002958699559	0.813953073112806\\
1.4142135623731	0.816587883219559\\
};
\addlegendentry{Average Power }

\end{axis}

\end{tikzpicture}%
\caption{Region of optimality of two points.}
\label{fig:Two-Four-Optimality}
\end{figure}
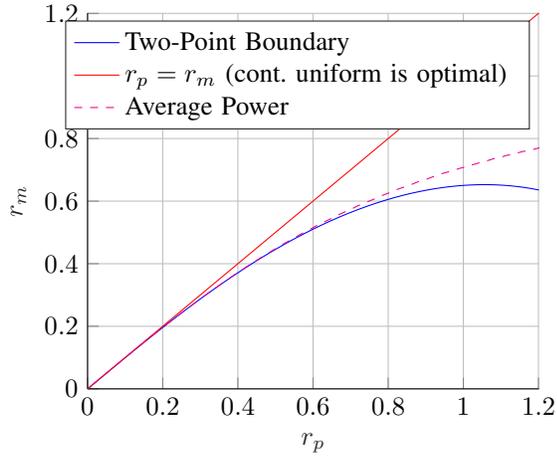

\begin{rem}
It is interesting to compare the solution in Theorem~\ref{thm:optimal_two_points} to having an average-power constraint. The capacity with an average-power constraint is
\begin{equation}
 \max_{ \bbE \left[  \| \bfX\|^2 \right] \le 1} I(\bfX; \mathsf{\bfD} \bfX+\bfZ), \label{eq:Average_Power_Capacity}
\end{equation}
where $\mathsf{\bfD}$ may be chosen as diagonal because the singular value decomposition of $\bfH$ and the isotropic pdf of $\bfZ$ permit this restriction. 
The capacity-achieving distribution is Gaussian with a diagonal covariance matrix with $\bbE[(X_1^\star)^2]=P_1$ and $\bbE[(X_2^\star)^2]=P_2$ and where $P_1$ and $P_2$ are given by the waterfilling solution
\begin{equation}
P_i=\max \left( \nu - \frac{1}{D_{ii}^2} ,0\right), \quad i \in \{1,2\},
\end{equation}
where $\nu$ is chosen so that $P_1+P_2=1$. To compare, let $D_{11}^2=r_p^2$ and $D_{22}^2=r_m^2$; We are interested in the regime where only one antenna is used (i.e., $P_2=0$).  The red dashed curve of Fig.~\ref{fig:Two-Four-Optimality} depicts the boundary of the region for which using only one antenna is optimal. Observe that the curves for the peak-power constraint and the average-power constraint at first follow each other but then they diverge. 
\end{rem}

\begin{rem}
Consider the $r_p,r_m$ for which a four-point distribution is optimal.  Unfortunately, finding a simple condition similar to \eqref{eq:Condition_for_optimality_two_point} for four points  does not seem possible. For example, such condition is not known even for $n=1$. Also, by symmetry, the optimal input distribution is
  \begin{equation}
P_{\bfX^\star}(\bfx) = \left \{ \begin{array}{ll} 
\frac{p}{2} &\bfx=(r_p,0)\\
\frac{p}{2} &\bfx=(-r_p,0)\\
\frac{1-p}{2} &\bfx=(0,r_m)\\
\frac{1-p}{2} &\bfx=(0,-r_m)\\
 \end{array}  \right.
\end{equation} 
for some $p \in [0,1]$. Because the distribution is parameterized by three parameters $(p,r_p ,r_m)$, it seems feasible to check the KKT conditions.  However, because of the multivariate nature of the problem, numerical methods suffer from instabilities, and we were unable to reliably compute $P_{\bfX^\star}$.  
\end{rem}

\subsection{On the Structure of $P_{\bfX^\star}$} 
\label{sec:on_discretness}

We next resolve the question posed in Section~\ref{sec:some_mysteries}, i.e., we show that the capacity-achieving input in the small peak-power regime is discrete with finitely many points if $r_p \neq r_m$ and is uniform on the circle with radius $r_p$ if $r_p=r_m$. 

\begin{theorem}\label{thm:discretness} Suppose $0\le r_m \le r_p<\sqrt{ 2}$. Then we have
\begin{itemize}
\item $\supp(P_{\bfX^\star})$ is a finite set for $r_m \neq  r_p$  or $r_m=r_p=0$;
\item $P_{\bfX^\star}$ is uniformly distributed on $ \partial \cE(r_p,r_p)$ for $0< r_m = r_p$, i.e., $\supp(P_{\bfX^\star})$ has infinite cardinality.    
\end{itemize} 
\end{theorem} 
\begin{proof}
See Section~\ref{proof:thm:discretness}. 
\end{proof} 

The next theorem shows that the points $(\pm r_p,0)$ are part of the optimal input distribution.

\begin{theorem}\label{thm:majoraxisoptimality}
	If $0\le r_m \le r_p$ then we have
	\begin{equation}
	\{( -r_p,0  ),(r_p,0)\} \subseteq \supp( P_{\bfX^\star}).
	\end{equation} 
\end{theorem}
\begin{proof}
	See Section~\ref{proof:thm:majoraxisoptimality}.
\end{proof}

\section{Proofs} 
\label{sec:proofs}
\subsection{Proof of  Theorem~\ref{thm:optimal_two_points}}
\label{proof:thm:optimal_two_points}

To check the optimality of $P_{\bfX^\star}$, we check the KKT conditions in Lemma~\ref{lem:KKT_conditions}. We have
\begin{equation}
    P_{\bfY^\star}(y_1,y_2) = P_{Y_1^\star}(y_1)P_{Y_2|X_2}(y_2|0)
\end{equation}
so the KL divergence simplifies to 
\begin{align}
&\kl{ P_{\bfY|\bfX}(\cdot | \bfx) }{  P_{\bfY^\star} } \nonumber \\&=\kl{ P_{Y_1|X_1}(\cdot | x_1) }{ P_{Y^\star_1} }+ \kl{ P_{Y_2|X_2}(\cdot | x_2) }{  P_{Y_2|X_2}(\cdot | 0) } \\
&=\kl{ P_{Y_1|X_1}(\cdot | x_1) }{  P_{Y^\star_1} }+ \frac{x_2^2}{2}\\
&=\kl{ P_{Y_1|X_1}(\cdot | x_1) }{  P_{Y^\star_1} }+ \frac{r_m^2 \left(1-\frac{x_1^2}{r_p^2}  \right)}{2},
\end{align} 
where the last step uses Lemma~\ref{thm:boundary_condition},  i.e., $\frac{x_1^2}{r_p^2}+\frac{x_2^2}{r_m^2}=1$. 
Now let
\begin{align}
f(x_1)&= \kl{ P_{Y_1|X_1}(\cdot | x_1) }{  P_{Y^\star_1} }+ \frac{r_m^2 \left(1-\frac{x_1^2}{r_p^2}  \right)}{2} \nonumber\\
&\quad- \kl{ P_{Y_1|X_1}(\cdot | r_p) }{  P_{Y^\star_1} },
\end{align}
where, due to symmetry, we have $\kl{ P_{Y_1|X_1}(\cdot | r_p) }{  P_{Y^\star_1} }=I(X_1^\star;Y_1^\star)$.  Furthermore, the KKT conditions imply that if $P_{X_1^\star}$ is optimal, then
\begin{equation}
{\rm supp}(P_{X_1^\star}) \subseteq \{ \text{ maxima of }  f(x_1) \text{ on } [-r_p,r_p] \}.
\end{equation} 
Therefore, we study the maxima of $f$ by finding the derivative of $f$, which is
\begin{align}
 f'(x_1)&=-  \frac{\rm d}{{\rm d}x_1} \bbE[ \log  f_{Y_1^\star}(Z+x_1)  ] -\frac{  r_m^2 x_1}{r_p^2}\\
&=x_1- \bbE[ \bbE[X_1^\star|Y_1^\star=Z+x_1] ] -\frac{  r_m^2 x_1}{r_p^2}  \label{eq:Using_Tweedy} \\
&= \left( 1-\frac{  r_m^2}{r_p^2}   \right) x_1- r_p  \bbE[ \tanh(r_p(Z+x_1)) ] \label{eq:Using_Expression_For_Optimal}\\
&= \bbE \left[  \left( 1-\frac{  r_m^2}{r_p^2}   \right) Z_{x_1}- r_p  \tanh(r_p Z_{x_1}) \right]\\
&= \bbE \left[  \eta(Z_{x_1}) \right], \label{eq:Definining_eta_function}
\end{align} 
where in \eqref{eq:Using_Tweedy} we have used Tweedie's formula~\cite{robbins1992empirical}
\begin{equation}
\bbE[X|Y=y]= y+\frac{\rm d}{{\rm d}y } \log f_Y(y).
\end{equation} 
In \eqref{eq:Using_Expression_For_Optimal} we used $ \bbE[X_1^\star|Y_1^\star=y]= r_p\tanh(r_p y) $; in \eqref{eq:Definining_eta_function} we defined $Z_{x_1} = Z+x_1$ and
\begin{equation}
\eta(y)=  \left( 1-\frac{  r_m^2}{r_p^2}   \right) y- r_p \tanh(r_p y). 
\end{equation} 

Now observe the following about the function $f'$:
\begin{enumerate}
\item at $x_1=0$, $f'$ has a sign change.  This follows from $f'(0)=0$ and that $f$ is even and, therefore, $f'$ is odd and there must be a sign change at $x_1=0$.
\item the function $f'$ can have at most three sign changes.  This follows from because $\eta(y)$ can have at most three sign changes and using Karlin's oscillation theorem (see, e.g., \cite{dytso2019capacity}) which states that the number of sign changes of $f'$ is upper-bounded by the number of sign changes of $\eta$. 
\item if  $r_m>0$ and the function $f'$ has exactly three sign changes, then sign changes must be in the order $-, +, -, +$.  This follows because there is a sign change at $x_1=0$ and $\lim_{ x_1 \to -\infty} f'(x_1)=-\infty$ and  $\lim_{ x_1 \to \infty} f'(x_1)=\infty$. Therefore, the  sign changes at $x_1 \ne 0$ must be from negative to positive, and the sign change at $x_1=0$ must be from positive to negative. 
\end{enumerate} 

In what follows, we focus on $x_1 \ge 0$. The above properties imply that only the following three cases are possible. 

{\it Case 1:} $f'$ is non-negative on $x_1 \ge 0$.  This implies that, the function $f$ is increasing and,  therefore,  is maximized at $x_1=r_p$. 

{\it Case 2:}  $f'$ is non-positive on $x_1 \ge 0$.  In this case, the function  $f$ is decreasing and is maximized at $x_1=0$.

{\it Case 3:} On $x_1>0$, $f'$  has a single sign change from negative to positive. Thus, $f$ has only one local minimum and no local maxima, and $f$ is maximized at $x_1=0$ or $x_1=r_p$.  

The above shows the function $f$ on $x_1 \ge 0$ can have maxima only at either  $x_1=0$ or $x_1=r_p$.  This  implies ${\rm supp}(P_{X_1}^\star)= \{-r_p,r_p \}$ if and only if 
  \begin{align}
  &f(0)<f(r_p) \nonumber\\  & \hspace{-0.1cm}\Leftrightarrow     \kl{ P_{Y_1|X_1}(\cdot | 0) }{  P_{Y^\star_1} }+ \frac{r_m^2 }{2}<  \kl{ P_{Y_1|X_1}(\cdot | r_p) }{  P_{Y^\star_1} }\\
&\hspace{-0.1cm} \Leftrightarrow   \frac{r_m^2}{2} < \hspace{-0.1cm} \int_{-\infty}^\infty \hspace{-0.2cm}\left(   \phi(y-r_p) - \phi(y) \right) \log \frac{1}{  \phi(y-r_p) + \phi(y+r_p)}   {\rm d} y
  \end{align}

\subsection{Proof of Theorem~\ref{thm:discretness}} 
\label{proof:thm:discretness}

Before proving the theorem, we present definitions and ancillary results.

\begin{definition} Let $\bfV \in \mathbb{R}^2$ have moment generating function $M_{\bfV}(\bft),\, \bft \in \mathbb{R}^2$. Then the cumulant generating function is
\begin{align}
K_{\bfV}(\bft)=\log( M_{\bfV}(\bft)  ), \,  \bft \in \mathbb{R}^2. 
\end{align} 
\end{definition} 

The following lemma is useful.  
\begin{lemma}   \label{lem:expression_for_log_pdf_Y} For $x_1 \in [-r_p,r_p]$ let 
\begin{equation}
f(x_1; P_{\bfX^\star})=  \bbE \left[   K_{\bar{\bfX}} \left(Z_1+x_1,Z_2+r_m \sqrt{1-\frac{x_1^2}{r_p^2}} \right) \right], 
\end{equation} 
where $\bar{\bfX}$ is distributed as $ \rmd P_{\bar{\bfX}}=  \frac{\rme^{ -\frac{  \|\bfx\|^2 }{2}} }{  \bbE[  \rme^{ -\frac{  \|\bfX^\star\|^2 }{2}}  ]} \rmd P_{\bfX^\star}$ (a tilted distribution).
The KKT conditions \eqref{eq:new_kkt_inequality}--\eqref{eq:new_KKT_equality} are equivalent to
\begin{align}
0 &\ge  \frac{x_1^2}{2} \left(1-\frac{r_m^2}{r_p^2} \right)- f(x_1; P_{\bfX^\star})-a,   \quad  x_1 \in [-r_p,r_p], \\
0 &=  \frac{x_1^2}{2} \left(1-\frac{r_m^2}{r_p^2} \right)- f(x_1; P_{\bfX^\star})-a,  \quad  x_1 \in \supp(P_{X_1^\star}). \label{eq:equalitzation_Equation}
\end{align}
where  $a=\sfC(r_p,r_m)-\log( 2 \pi) +\log \left( \bbE \left[  \rme^{ -\frac{  \|\bfX^\star\|^2 }{2}}  \right ]  \right)  + h(\bfZ)- 1-\frac{r_m^2}{2}$.  

\end{lemma} 
\begin{proof} See the supplementary material in Appendix~\ref{app:lem:expression_for_log_pdf_Y}.
\end{proof} 

\begin{definition} For  $z \in \mathbb{C}$ the principle square root function is
\begin{equation}
z^{ \frac{1}{2}}= |z|^{ \frac{1}{2} } \rme^{ \frac{i  \, {\rm arg}(z)}{2}}. 
\end{equation} 
\end{definition} 

\begin{rem}
The principle square root function is holomorphic on $ \mathbb{C} \setminus  \{ z:  {\rm Re}(z) \le 0, {\rm Im}(z)=0 \}  $ \cite{bak2010complex}.    
\end{rem}

\begin{lemma}\label{lem:complex_ext_of_cumulant} The function  $ \bfx  \mapsto \bbE \left[   K_{\bar{\bfX}} \left(\bfZ+\bfx \right) \right]$ has the complex extension 
\begin{equation}
  \bbE_{\bfZ} \left[   K_{\bar{\bfX}} \left(\bfZ+\bfz \right) \right]=   \bbE_{\bfZ} \left[  \log \bbE_{ \bar{\bfX}} \left [ \rme^{ (\bfZ+\bfz)^T \bar{\bfX}}  \right] \right], \,  \bfz \in \mathbb{C}^2. \label{eq:extention of K_X}
\end{equation}
Moreover, the function in \eqref{eq:extention of K_X} is an entire function, i.e., analytic on $\mathbb{C}^2$.  Consequently, the function $x_1 \mapsto f(x_1; P_{\bfX^\star})$ has a complex  extension
\begin{align}
&\tilde{f}(z; P_{\bfX^\star}) \notag\\
&=   \bbE_{\bfZ} \left[      K_{\bar{\bfX}} \left(Z_1+z,Z_2+r_m  \sqrt{  \left| 1-\frac{z^2}{r_p^2} \right| }  \rme^{ \frac{i  \, {\rm arg} \left(1-\frac{z^2}{r_p^2} \right)}{2}} \right)    \right]
\end{align} 
which is analytic on $z \in  \mathbb{C} \setminus  \{ z:   |{\rm Re}(z)| \ge r_p, {\rm Im}(z)=0 \}$. 
\end{lemma}
\begin{proof}
See the supplementary material in Appendix~\ref{app:lem:complex_ext_of_cumulant}.
\end{proof} 

The next lemma provides an upper bound on the peak power of $\left |   \tilde{f}(z; P_{\bfX^\star}) \right| $ along the imaginary axis.

\begin{lemma}\label{lem:Bound_on_abs_of_f}  For all $u \in \mathbb{R}$
\begin{equation}
\left |   \tilde{f}(i u; P_{\bfX^\star}) \right|  \le   \frac{ r_p^2}{2} +  \frac{ r_m^2}{2} +r_m^2 \sqrt{ 1+\frac{u^2}{r_p^2} }         +2\pi. 
\end{equation} 
\end{lemma} 
\begin{proof}
See the supplementary material in Appendix~\ref{app:lem:Bound_on_abs_of_f}. 
\end{proof}

\emph{Proof of Theorem~\ref{thm:discretness}:} Suppose $\supp(P_{X_1^\star})$ has infinite cardinality. Starting with one of the KKT equations in \eqref{eq:equalitzation_Equation}, we have
\begin{align}
0 &=  \frac{x_1^2}{2} \left(1-\frac{r_m^2}{r_p^2} \right)- f(x_1; P_{\bfX^\star})-a,  \,  x_1 \in \supp(P_{X_1^\star})\\
&\Rightarrow    \, 
0 =  \frac{x_1^2}{2} \left(1-\frac{r_m^2}{r_p^2} \right)- f(x_1; P_{\bfX^\star})-a,  \,  x_1 \in (-r_p,r_p) \label{eq:Identity_for_analytic}\\
&\Rightarrow    \,  0 =  \frac{z^2}{2} \left(1-\frac{r_m^2}{r_p^2} \right)- \tilde{f}(z; P_{\bfX^\star})-a, \notag\\
&  \quad \quad \quad \,  z \in  \mathbb{C} \setminus  \left\{ z:   |{\rm Re}(z)| \ge r_p, {\rm Im}(z)=0 \right\}\label{eq:Complex_analytic_Extension}\\
&\Rightarrow  0 = -  \frac{u^2}{2} \left(1-\frac{r_m^2}{r_p^2} \right)- \tilde{f}(iu; P_{\bfX^\star})-a,  \, u\in \mathbb{R} , \label{eq:Finnal_ mplicaionion} 
\end{align} 
where \eqref{eq:Identity_for_analytic} follows  because (Bolzano-Weierstrass Theorem with the Identity Theorem) an analytic function that is equal to zero  on a closed bounded set of infinite cardinality is equal to zero everywhere on its domain;  and \eqref{eq:Complex_analytic_Extension} follows because $\tilde{f}(z; P_{\bfX^\star})$ is the analytic extension of $f(x; P_{\bfX^\star})$. 
Now, Lemma~\ref{lem:Bound_on_abs_of_f} shows that  $ | \tilde{f}(iu; P_{\bfX^\star}) | \in O(u)$, together with \eqref{eq:Finnal_ mplicaionion}, leads to a contradiction.  

\subsection{Proof of Theorem~\ref{thm:majoraxisoptimality}  } 
\label{proof:thm:majoraxisoptimality}

By symmetry of the channel (see Proposition~\ref{prop:symmetry}), we may focus on the case $x_1>0$.

\begin{lemma}\label{lemma:increasingDx1}
	Suppose $0\le r_m \le r_p$ and $\Pr( X_1>t)=0$ for some $0<t<r_p$. Then for $x_1>t$, we have
	\begin{equation}
		\frac{\rmd}{\rmd x_1}\kl{ P_{\bfY|\bfX}(\cdot | (x_1,x_2)) }{  P_{\bfY} } > 0, \qquad \forall x_2\in \mathbb{R}.
	\end{equation}
\end{lemma}
\begin{proof}
	For all $x_2 \in \mathbb{R}$, we can write
	\begin{align}
		&\frac{\rmd}{\rmd x_1}\kl{ P_{\bfY|\bfX}(\cdot | (x_1,x_2)) }{  P_{\bfY} }  \notag\\
  &= - \frac{d}{dx_1} \bbE \left[ \log f_{\bfY}(\bfx+\bfZ)\right] \\
		&=-\bbE \left[\nabla \log f_{\bfY}(\bfx+\bfZ)\right]^\dagger \cdot \frac{d}{dx_1}\bfx \\
		&=-\bbE \left[ \bbE\left[\bfX | \bfY = \bfx+\bfZ\right]-\bfx+\bfZ  \right]^\dagger \cdot \left[\begin{array}{c}
		1 \\ 0
		\end{array}  \right] \label{eq:use_Tweedie_vector} \\
		&=-\bbE \left[ \bbE\left[X_1 | \bfY = \bfx+\bfZ\right] \right] +x_1 \\
		&\ge -t+x_1 \label{eq:no_mass_after_t},
	\end{align}
	where~\eqref{eq:use_Tweedie_vector} follows from Tweedie's formula~\cite{robbins1992empirical} and~\eqref{eq:no_mass_after_t} due to $\Pr( X_1>t)=0$. Finally, combine~\eqref{eq:no_mass_after_t} with $x_1 > t$.
\end{proof}
Arguing by contradiction, let $t<r_p$ be the largest value of the $\supp(P_{X_1^\star})$. Then, Lemma~\ref{lemma:increasingDx1} implies
\begin{equation}
    \kl{ P_{\bfY|\bfX}(\cdot | (x_1,x_2)) }{  P_{\bfY^\star} } > \kl{ P_{\bfY|\bfX}(\cdot | (t,x_2)) }{  P_{\bfY^\star} }
\end{equation}
for all $x_1>t$ and $x_2 \in \bbR$. This would violate the KKT condition~\eqref{eq:original_KKT_inequality}. 
As a consequence, we necessarily have $(r_p,0) \in \supp(P_{\bfX^\star})$ and, by symmetry, $(-r_p,0) \in \supp(P_{\bfX^\star})$.

\section{Discussion and Conclusion}
\label{sec:disc_conc}
We studied  the capacity-achieving distribution of $2\times 2$ MIMO channels subject to a peak-power constraint. We characterized conditions for the optimality of a two-point distribution and, by transforming the problem into a one-dimensional problem, we derived sufficient conditions for the discreteness of capacity-achieving distributions. 

There are several interesting directions for future work. For example, an ambitious goal is extending the results to an any number of antennas. 
It would also be intriguing to address questions in the $2\times 2 $ setting. For instance, let  $P_{\bfX^\star_{r_p,r_m}}$ be the optimal input distribution for the specified values of $r_p$ and $r_m$. Since capacity is a concave function of input constraint \cite{shannon1959coding} (see also \cite[Ch.~3.3]{el2011network}) and since concave functions are continuous  \cite{webster1994convexity}, we have
\begin{align}
&\lim_{r_m \to r_p}  \sfC(r_p,r_m)= \sfC(r_p,r_p)  \notag\\
&\quad \Longleftrightarrow  \lim_{r_m \to r_p}  I(\bfX^\star_{r_p,r_m};  \bfX^\star_{r_p,r_m}+\bfZ)= I(\bfX^\star_{r_p};  \bfX^\star_{r_p}+\bfZ),
\end{align}
where $\bfX^\star_{r_p} =\bfX^\star_{r_p,r_p} $. 
Note that $\bfX^\star_{r_p}$ is uniformly distributed on the circle of radius $r_p$. Furthermore, by the continuity of mutual information in the distribution \cite{wu2011functional}, we have
\begin{equation}
\lim_{r_m \to r_p}  \bfX^\star_{r_p,r_m} \stackrel{d}{=} \bfX^\star_{r_p}.
\end{equation} 
Note that $P_{\bfX^\star_{r_p}}$ has infinite support while $P_{\bfX^\star_{r_p,r_m}}$ has finite support for all $r_m <r_p \le \sqrt{2}$ by Theorem~\ref{thm:discretness}.  Thus, an interesting open problem is the rate at which the cardinality of the set $\supp(P_{\bfX^\star_{r_p,r_m}})$ approaches infinity as $r_m \to r_p$.

\bibliographystyle{IEEEtran}
\bibliography{biblio}

\begin{thebibliography}{10}
\providecommand{\url}[1]{#1}
\csname url@samestyle\endcsname
\providecommand{\newblock}{\relax}
\providecommand{\bibinfo}[2]{#2}
\providecommand{\BIBentrySTDinterwordspacing}{\spaceskip=0pt\relax}
\providecommand{\BIBentryALTinterwordstretchfactor}{4}
\providecommand{\BIBentryALTinterwordspacing}{\spaceskip=\fontdimen2\font plus
\BIBentryALTinterwordstretchfactor\fontdimen3\font minus
  \fontdimen4\font\relax}
\providecommand{\BIBforeignlanguage}[2]{{%
\expandafter\ifx\csname l@#1\endcsname\relax
\typeout{** WARNING: IEEEtran.bst: No hyphenation pattern has been}%
\typeout{** loaded for the language `#1'. Using the pattern for}%
\typeout{** the default language instead.}%
\else
\language=\csname l@#1\endcsname
\fi
#2}}
\providecommand{\BIBdecl}{\relax}
\BIBdecl

\bibitem{smith1971information}
J.~G. Smith, ``The information capacity of amplitude-and variance-constrained
  scalar {G}aussian channels,'' \emph{Inf. and Control}, vol.~18, no.~3, pp.
  203--219, 1971.

\bibitem{sharma2010transition}
N.~Sharma and S.~Shamai, ``Transition points in the capacity-achieving
  distribution for the peak-power limited {AWGN} and free-space optical
  intensity channels,'' \emph{Prob. Inf. Transm.}, vol.~46, pp. 283--299, 2010.

\bibitem{shamai1995capacity}
S.~Shamai and I.~Bar-David, ``The capacity of average and peak-power-limited
  quadrature {G}aussian channels,'' \emph{IEEE Trans. Inf. Theory}, vol.~41,
  no.~4, pp. 1060--1071, 1995.

\bibitem{dytsoEstPerspective}
A.~Dytso, M.~Al, H.~V. Poor, and S.~Shamai~Shitz, ``On the capacity of the peak
  power constrained vector {G}aussian channel: An estimation theoretic
  perspective,'' \emph{IEEE Trans. Inf. Theory}, vol.~65, no.~6, pp.
  3907--3921, 2019.

\bibitem{dytso2019capacity}
A.~Dytso, S.~Yagli, H.~V. Poor, and S.~S. Shamai, ``The capacity achieving
  distribution for the amplitude constrained additive {G}aussian channel: {A}n
  upper bound on the number of mass points,'' \emph{IEEE Trans. Inf. Theory},
  vol.~66, no.~4, pp. 2006--2022, 2019.

\bibitem{capBoundSISO_amp}
A.~Thangaraj, G.~Kramer, and G.~Böcherer, ``Capacity bounds for discrete-time,
  amplitude-constrained, additive white {G}aussian noise channels,'' \emph{IEEE
  Trans. Inf. Theory}, vol.~63, no.~7, pp. 4172--4182, 2017.

\bibitem{rassouli2016-2}
B.~Rassouli and B.~Clerckx, ``An upper bound for the capacity of
  amplitude-constrained scalar {AWGN} channel,'' \emph{IEEE Commun. Lett.},
  vol.~20, no.~10, pp. 1924--1926, 2016.

\bibitem{mckellips2004simple}
A.~L. McKellips, ``Simple tight bounds on capacity for the peak-limited
  discrete-time channel,'' in \emph{IEEE Int. Symp. Inf. Theory}, 2004, pp.
  348--348.

\bibitem{dytso2017generalized}
A.~Dytso, M.~Goldenbaum, H.~V. Poor, and S.~S. Shitz, ``A generalized
  {O}zarow-{W}yner capacity bound with applications,'' in \emph{IEEE Int. Symp.
  Inf. Theory}.\hskip 1em plus 0.5em minus 0.4em\relax IEEE, 2017, pp.
  1058--1062.

\bibitem{e21020200}
\BIBentryALTinterwordspacing
A.~Dytso, M.~Goldenbaum, H.~V. Poor, and S.~Shamai~(Shitz), ``Amplitude
  constrained {MIMO} channels: Properties of optimal input distributions and
  bounds on the capacity,'' \emph{Entropy}, vol.~21, no.~2, 2019. [Online].
  Available: \url{https://www.mdpi.com/1099-4300/21/2/200}
\BIBentrySTDinterwordspacing

\bibitem{eisen2023capacity}
J.~Eisen, R.~R. Mazumdar, and P.~Mitran, ``Capacity-achieving input
  distributions of additive vector {G}aussian noise channels: Even-moment
  constraints and unbounded or compact support,'' \emph{Entropy}, vol.~25,
  no.~8, p. 1180, 2023.

\bibitem{chan2005capacity}
T.~H. Chan, S.~Hranilovic, and F.~R. Kschischang, ``Capacity-achieving
  probability measure for conditionally {G}aussian channels with bounded
  inputs,'' \emph{IEEE Trans. Inf. Theory}, vol.~51, no.~6, pp. 2073--2088,
  2005.

\bibitem{rassouli2016capacity}
B.~Rassouli and B.~Clerckx, ``On the capacity of vector {G}aussian channels
  with bounded inputs,'' \emph{IEEE Trans. Inf. Theory}, vol.~62, no.~12, pp.
  6884--6903, 2016.

\bibitem{favano2021capacity}
A.~Favano, M.~Ferrari, M.~Magarini, and L.~Barletta, ``The capacity of the
  amplitude-constrained vector {G}aussian channel,'' in \emph{IEEE Int. Symp.
  Inf. Theory}.\hskip 1em plus 0.5em minus 0.4em\relax IEEE, 2021, pp.
  426--431.

\bibitem{capDuman}
A.~El~Moslimany and T.~M. Duman, ``On the capacity of multiple-antenna systems
  and parallel {G}aussian channels with amplitude-limited inputs,'' \emph{IEEE
  Trans. Commun.}, vol.~64, no.~7, pp. 2888--2899, 2016.

\bibitem{genMIMOcapBoundsGlob}
A.~Dytso, M.~Goldenbaum, S.~Shamai, and H.~V. Poor, ``Upper and lower bounds on
  the capacity of amplitude-constrained {MIMO} channels,'' in \emph{IEEE Global
  Commun. Conf.}, 2017, pp. 1--6.

\bibitem{favano2022sphere}
A.~Favano, M.~Ferrari, M.~Magarini, and L.~Barletta, ``A sphere packing bound
  for vector {G}aussian fading channels under peak amplitude constraints,''
  \emph{IEEE Trans. Inf. Theory}, vol.~69, no.~1, pp. 238--250, 2022.

\bibitem{hatsell1971some}
C.~Hatsell and L.~Nolte, ``Some geometric properties of the likelihood ratio
  (corresp.),'' \emph{IEEE Trans. Inf. Theory}, vol.~17, no.~5, pp. 616--618,
  1971.

\bibitem{dytso2022conditional}
A.~Dytso, H.~V. Poor, and S.~S. Shamai, ``Conditional mean estimation in
  {G}aussian noise: {A} meta derivative identity with applications,''
  \emph{IEEE Trans. Inf. Theory}, vol.~69, no.~3, pp. 1883--1898, 2022.

\bibitem{stein2010complex}
E.~M. Stein and R.~Shakarchi, \emph{Complex Analysis}.\hskip 1em plus 0.5em
  minus 0.4em\relax Princeton University Press, 2010, vol.~2.

\bibitem{robbins1992empirical}
H.~E. Robbins, ``An empirical {B}ayes approach to statistics,'' in
  \emph{Breakthroughs in Statistics: Foundations and basic theory}.\hskip 1em
  plus 0.5em minus 0.4em\relax Springer, 1992, pp. 388--394.

\bibitem{bak2010complex}
J.~Bak, D.~J. Newman, and D.~J. Newman, \emph{Complex Analysis}.\hskip 1em plus
  0.5em minus 0.4em\relax Springer, 2010, vol.~8.

\bibitem{shannon1959coding}
C.~E. Shannon \emph{et~al.}, ``Coding theorems for a discrete source with a
  fidelity criterion,'' \emph{IRE Nat. Conv. Rec}, vol.~4, no. 142-163, p.~1,
  1959.

\bibitem{el2011network}
A.~El~Gamal and Y.-H. Kim, \emph{Network information theory}.\hskip 1em plus
  0.5em minus 0.4em\relax Cambridge university press, 2011.

\bibitem{webster1994convexity}
R.~Webster, \emph{Convexity}.\hskip 1em plus 0.5em minus 0.4em\relax Oxford
  university press, 1994.

\bibitem{wu2011functional}
Y.~Wu and S.~Verd{\'u}, ``Functional properties of minimum mean-square error
  and mutual information,'' \emph{IEEE Trans. Inf. Theory}, vol.~58, no.~3, pp.
  1289--1301, 2011.

\end{thebibliography}

\clearpage

\newpage 
\begin{appendices}

\section{Proof of the Equivalence of \eqref{eq:Original_cap} and \eqref{eq:Capacity_Def} }
\label{sec:proof_of_equivalence}
Let $\bfU \bfQ \bfV^\sfT$ be the singular-value decomposition of $\bfH$, i.e., $\bfU$ and $\bfV$ are unitary and $\bfQ$ is diagonal. We have 
\begin{align}
 \max_{\bfX :  \| \bfX\| \le 1} & I \left(\bfX; \mathsf{\bfH} \bfX+\bfZ \right) \nonumber \\
 &\overset{(a)}= \max_{\bfX :  \| \bfX\| \le 1} I \left(\bfV^\sfT \bfX; \bfU^\sfT( \mathsf{\bfH} \bfX+\bfZ) \right) \nonumber \\
  &\overset{(b)}= \max_{\bfX :  \| \bfX\| \le 1} I \left(\bfV^\sfT\bfX; \bfQ \bfV^\sfT \bfX+\bfZ) \right)  \nonumber \\
    &\overset{(c)}= \max_{ \tilde{\bfX} :  \|   \tilde{\bfX}\| \le 1} I \left(\tilde{\bfX}; \bfQ \tilde{\bfX}+\bfZ \right) \nonumber \\
     &\overset{(d)}= \max_{ \bar{\bfX} :  \|  \bfQ^{-1} \bar{\bfX}\| \le 1} I \left(\bar{\bfX}; \bar{\bfX}+\bfZ \right) \nonumber \\
     &\overset{(e)}= \max_{ \bar{\bfX} : \, \frac{\bar{X}_1^2}{\sigma_1^2} +\frac{\bar{X}_2^2}{\sigma_2^2}  \le 1  } I \left(\bar{\bfX}; \bar{\bfX}+\bfZ \right)  \nonumber \\
     &= \sfC(\sigma_1,\sigma_2),
 \end{align}
 where step $(a)$ follows because mutual information is invariant under one-to-one transformations; step $(b)$ follows because  $\bfU^\sfT\bfZ$ and $\bfZ$ have the same distribution; step $(c)$ follows because the Euclidean norm is invariant under unitary transformation; in step $(d)$ if one of the diagonal elements in $\bfQ$ is zero we set the corresponding diagonal value of the inverse  to positive infinity; and in step $(e)$ $\sigma_1$ and $\sigma_2$ are the singular values of $\bfQ$. 

\section{Derivatives of Information Density}
Consider the information density
\begin{equation}
    i(\bfx; P_{\bfY}) = \kl{ P_{\bfY|\bfX}(\cdot | \bfx) }{  P_{\bfY} }
\end{equation}
evaluated on the upper-half of the ellipse $\del\cE(r_p,r_m) \cap \{\bfx : \: x_2\ge 0\}$, i.e., for $\bfx \in \del\cE^{+}$ where
\begin{equation}
    \del\cE^{+} = \left\{\left(x_1, r_m \sqrt{1-\frac{x_1^2}{r_p^2}}\right):\: x_1 \in [-r_p, r_p] \right\}.
\end{equation}
\begin{lemma} \label{lem:der_info_den_on_ellipse}
    For $\bfx \in \del\cE^{+}$, we have
    \begin{align}
        &\frac{\rmd}{\rmd x_1} i(\bfx; P_{\bfY}) \notag\\
        &=-\bbE \left[ \bbE\left[X_1-\frac{r_m}{r_p^2 \sqrt{1-\frac{x_1^2}{r_p^2}}}x_1 X_2 | \bfY = \bfx+\bfZ\right]  \right] \notag\\
        &\quad +x_1\left(1-\frac{r_m^2}{r_p^2}\right)
    \end{align}
    and
    \begin{align}
        \frac{\rmd^2}{\rmd x_1^2} i(\bfx; P_{\bfY})
        &= \bbE\left[\dot{\bfx}^\dagger (\sfI-\mathsf{Var} (\bfX| \bfY=\bfx+\bfZ)  ) \dot{\bfx} \right] \notag\\
        & \quad +\bbE \left[ \bbE\left[\frac{r_m}{r_p^2 \left(1-\frac{x_1^2}{r_p^2}\right)^{\frac{3}{2}}}X_2 | \bfY = \bfx+\bfZ\right] \right] \notag\\
        & \quad-\frac{r_m^2}{r_p^2 \left(1-\frac{x_1^2}{r_p^2}\right)}.
    \end{align}
\end{lemma}
\begin{proof}
First, let 
    \begin{equation}
        \dot{\bfx} = \frac{\rmd }{\rmd x_1}\bfx =  \left[\begin{array}{c}
		1 \\ -\frac{r_m}{r_p^2 \sqrt{1-\frac{x_1^2}{r_p^2}}}x_1
		\end{array}  \right]
    \end{equation}
    and
    \begin{equation}
        \ddot{\bfx} = \frac{\rmd^2 }{\rmd x_1^2}\bfx =  \left[\begin{array}{c}
		0 \\ -\frac{r_m}{r_p^2 \left(1-\frac{x_1^2}{r_p^2}\right)^{\frac{3}{2}}}
		\end{array}  \right].
    \end{equation}
    
    For the first derivative, we have:
    \begin{align}
        -&\frac{\rmd}{\rmd x_1} i(\bfx; P_{\bfY}) \notag\\
        &=  \frac{\rmd}{\rmd x_1} \bbE \left[ \log f_{\bfY}(\bfx+\bfZ)\right] \nonumber \\
		&\overset{(a)}=\bbE \left[\nabla \log f_{\bfY}(\bfx+\bfZ)\right]^\dagger \cdot \frac{\rmd}{\rmd x_1}(\bfx+\bfZ) \nonumber \\
		&\overset{(b)}=\bbE \left[ \bbE\left[\bfX | \bfY = \bfx+\bfZ\right]-\bfx-\bfZ  \right]^\dagger \cdot \dot{\bfx} \nonumber \\
        &= \bbE \left[ \bbE\left[X_1-\frac{r_m}{r_p^2 \sqrt{1-\frac{x_1^2}{r_p^2}}}x_1 X_2 | \bfY = \bfx+\bfZ\right]  \right]  \notag\\
        & \quad -x_1\left(1-\frac{r_m^2}{r_p^2}\right),
    \end{align}
    where step $(a)$ holds by applying the chain rule for the derivative of composite multivariable functions; and step $(b)$ follows from Tweedie's formula. 

    For the second derivative, we can write
    \begin{align}
        -&\frac{\rmd^2}{\rmd x_1^2} i(\bfx; P_{\bfY}) \notag\\
        &= \frac{\rmd^2 }{\rmd x_1^2} \bbE \left[ \log f_{\bfY}(\bfx+\bfZ)\right] \nonumber \\
        &\overset{(a)}=\bbE\left[\dot{\bfx}^\dagger \sfH[\log f_{\bfY}](\bfx+\bfZ) \dot{\bfx} \right]+\bbE \left[\nabla \log f_{\bfY}(\bfx+\bfZ)\right]^\dagger \cdot \ddot{\bfx} \nonumber \\
        &\overset{(b)}= \bbE\left[\dot{\bfx}^\dagger (\mathsf{Var} (\bfX| \bfY=\bfx+\bfZ) - \sfI) \dot{\bfx} \right] \notag\\
        &\quad +\bbE \left[ \bbE\left[\bfX | \bfY = \bfx+\bfZ\right]-\bfx  \right]^\dagger \cdot \ddot{\bfx}  \nonumber \\
        &= \bbE\left[\dot{\bfx}^\dagger (\mathsf{Var} (\bfX| \bfY=\bfx+\bfZ) - \sfI) \dot{\bfx} \right] \notag\\
        & \quad +\bbE \left[ \bbE\left[-\frac{r_m}{r_p^2 \left(1-\frac{x_1^2}{r_p^2}\right)^{\frac{3}{2}}}X_2 | \bfY = \bfx+\bfZ\right] \right] \notag\\
        & \quad +\frac{r_m^2}{r_p^2 \left(1-\frac{x_1^2}{r_p^2}\right)},
    \end{align}
    where step $(a)$ holds by the chain rule for the second derivative of composite multivariable functions; and step $(b)$ follows by the Hatsel-Nolte identity and Tweedie's formula.
\end{proof}

\section{Proof of Lemma~\ref{lem:expression_for_log_pdf_Y} }
\label{app:lem:expression_for_log_pdf_Y}

We first re-write the log of the pdf of $\bfY^\star$
\begin{align}
&\log f_{\bfY^\star}(\bfy) \notag\\
&=   \log \left( \bbE \left[  \frac{1}{2 \pi }\rme^{ -\frac{ \|\bfy\|^2 - 2 \bfy^T \bfX^\star +\|\bfX^\star
\|^2 }{2}}  \right] \right) \nonumber \\
&=-\log( 2 \pi) - \frac{\|\bfy\|^2}{2}+   \log \left( \bbE \left [ \rme^{ -\frac{  - 2\bfy^T \bfX^\star+\|\bfX^\star\|^2 }{2}}  \right ] \right) \nonumber\\
&=-\log( 2 \pi) - \frac{\|\bfy\|^2}{2}+  \log\left( \bbE \left[  \rme^{ -\frac{  \|\bfX^\star\|^2 }{2}}  \right ] \right) \notag\\
& \quad + \log \left( \bbE \left [ \rme^{\bfy^T  \bar{\bfX} } \right ] \right) \nonumber \\
&=-\log( 2 \pi) - \frac{\|\bfy\|^2}{2}+  \log \left( \bbE \left[  \rme^{ -\frac{  \|\bfX^\star\|^2 }{2}} \right  ]  \right)+ K_{\bar{\bfX}}(\bfy) \label{eq:Representation_of_f_Y} 
\end{align} 
where $ K_{\bar{\bfX}}(\bfy)$ is the cumulant generating function of $\bar{\bfX}$ and  $\bar{\bfX}$ is distributed according to $ \rmd P_{\bar{\bfX}}=  \frac{\rme^{ -\frac{  \|\bfx\|^2 }{2}} }{  \bbE[  \rme^{ -\frac{  \|\bfX^\star\|^2 }{2}}  ]} \rmd P_{\bfX^\star}$.  

In what follows, the symmetry in Proposition~\ref{prop:symmetry} implies that we can consider only $x_2 \ge 0$ and restrict attention to the upper half of the ellipse
\begin{equation}
 \cE^{+} \equiv \partial  \cE(r_p,r_m) \cap \{ \bfx:  x_2 \ge 0 \} . 
\end{equation}

Now  for $\bfx \in  \cE^{+}$, the relative entropy can be rewritten as 
\begin{align}
&\kl{ P_{\bfY|\bfX}(\cdot | \bfx) }{  P_{\bfY^\star} } \nonumber\\
&= -\bbE \left[  \log f_{\bfY^\star}(\bfY)|\bfX=\bfx \right] - h(\bfZ) \nonumber\\
&\overset{(a)}=\log( 2 \pi)- \log \left( \bbE \left[  \rme^{ -\frac{  \|\bfX^\star\|^2 }{2}}  \right ]  \right)  - h(\bfZ) \notag\\
& \quad + \bbE \left[  \frac{\|\bfY\|^2}{2} |\bfX=\bfx \right]-  \bbE \left[   K_{\bar{\bfX}}(\bfY) |\bfX=\bfx \right] \nonumber\\
&\overset{(b)}=\log( 2 \pi)- \log \left( \bbE \left[  \rme^{ -\frac{  \|\bfX^\star\|^2 }{2}}  \right ]  \right)  - h(\bfZ) \notag\\
& \quad +  \frac{2+\|\bfx\|^2}{2}-  \bbE \left[   K_{\bar{\bfX}}(Z_1+x_1,Z_2+x_2)  \right] \nonumber \\
&=\log( 2 \pi)- \log \left( \bbE \left[  \rme^{ -\frac{  \|\bfX^\star\|^2 }{2}}  \right ]  \right)  - h(\bfZ) \notag\\
&\quad + 1 +\frac{r_m^2 }{2}+ \frac{x_1^2}{2} \left(1-\frac{r_m^2}{r_p^2} \right) \notag\\
& \quad -  \bbE \left[   K_{\bar{\bfX}} \left(Z_1+x_1,Z_2+r_m \sqrt{1-\frac{x_1^2}{r_p^2}} \right) \right] \nonumber \\
&=  -f(x_1; P_{\bfX^\star})-a+\sfC(r_p,r_m) \label{eq:using_parametrization_of_ellipse}
\end{align} 
where in step $(a)$ we have used the identity in \eqref{eq:Representation_of_f_Y};  step $(b)$ follows because $\bfY| \bfX=\bfx$ is Gaussian with mean $\bfx$; and step $(c)$ follows because $\cE^{+}$ can be described only in terms of $x_1$:
\begin{align*}
 \cE^{+}&\equiv \left \{ (x_1,x_2):   x_2=r_m \sqrt{1-\frac{x_1^2}{r_p^2}} , \,  x_1 \in [-r_p,r_p]  \right \} .
\end{align*} 

Combining \eqref{eq:using_parametrization_of_ellipse} with the KKT conditions in \eqref{eq:new_kkt_inequality}-\eqref{eq:new_KKT_equality} and defining $a=\sfC(r_p,r_m)-\log( 2 \pi) +\log \left( \bbE \left[  \rme^{ -\frac{  \|\bfX^\star\|^2 }{2}}  \right ]  \right)  + h(\bfZ)- 1-\frac{r_m^2}{2}$ concludes the proof.  

\section{Proof of Lemma~\ref{lem:complex_ext_of_cumulant}}
\label{app:lem:complex_ext_of_cumulant}

Using \eqref{eq:using_parametrization_of_ellipse}, we can write 
\begin{align}
&\kl{ P_{\bfY|\bfX}(\cdot | \bfx) }{  P_{\bfY^\star} } \notag\\
&=c+  \frac{2+\|\bfx\|^2}{2}-  \bbE \left[   K_{\bar{\bfX}}(Z_1+x_1,Z_2+x_2)  \right] 
\end{align}
for some constant $c$ independent of $\bfx$. 
Therefore, $\bfx \mapsto  \bbE \left[   K_{\bar{\bfX}}(Z_1+x_1,Z_2+x_2)  \right] $ has a complex extension which is an entire  function if and only if $\bfx \mapsto \kl{ P_{\bfY|\bfX}(\cdot | \bfx) }{  P_{\bfY^\star} }$ has complex extension which is  an entire function.  The fact that $\bfx \mapsto \kl{ P_{\bfY|\bfX}(\cdot | \bfx) }{  P_{\bfY^\star} }$ has a complex analytic extension to $\mathbb{C}^2$ has been established in several works \cite{chan2005capacity}.  

Next, observe that $\tilde{f}(z; P_{\bfX^\star})$ is composition of the complex extensions  $\bfz \mapsto \bbE \left[   K_{\bar{\bfX}}(Z_1+z_1,Z_2+z_2)  \right] $ and $z \mapsto \left (z,\sqrt{  \left| 1-\frac{z^2}{r_p^2} \right| }  \rme^{ \frac{i  \, {\rm arg} \left(1-\frac{z^2}{r_p^2} \right)}{2}} \right)$. Note that the former function is entire and the  latter function is analytic on $z \in  \mathbb{C} \setminus  \{ z:   |{\rm Re}(z)| \ge r_p, {\rm Im}(z)=0 \}$. Therefore,  the proof that $\tilde{f}(z; P_{\bfX^\star})$  is analytic on  $z \in  \mathbb{C} \setminus  \{ z:   |{\rm Re}(z)| \ge r_p, {\rm Im}(z)=0 \}$ follows because a composition of analytic functions is analytic \cite{bak2010complex}.

\section{Proof of Lemma~\ref{lem:Bound_on_abs_of_f}}
\label{app:lem:Bound_on_abs_of_f}
The proof of the bound proceeds as follows: 
\begin{align}
&\left| \tilde{f}(i u; P_{\bfX^\star}) \right|
  \le   \bbE \left[   \left|   K_{\bar{\bfX}} \left(Z_1+iu,Z_2+r_m \sqrt{ 1+\frac{u^2}{r_p^2} } \right)  \right|   \right]  \nonumber\\
  &\overset{(a)}\le   \bbE \left[   \log    \left| \bbE \left[   \rme^{  (Z_1+iu) X_1^\star +  (Z_2+r_m \sqrt{ 1+\frac{u^2}{r_p^2} }) X_2^\star   }\right]\right|    \right] +2\pi  \nonumber\\
    &\overset{(b)}\le   \bbE \left[   \log    \bbE \left[   \left|  \rme^{  (Z_1+iu) X_1^\star +  (Z_2+r_m \sqrt{ 1+\frac{u^2}{r_p^2} }) X_2^\star   }  \right|  \right]   \right] +2\pi \nonumber\\
    &=   \bbE \left[   \log    \bbE \left[     \rme^{  Z_1 X_1^\star +  (Z_2+r_m \sqrt{ 1+\frac{u^2}{r_p^2} }) X_2^\star   }  \right]   \right] +2\pi \nonumber \\
    &\overset{(c)}\le    \log    \bbE_{Z_1,Z_2} \left[  \bbE_{\bfX^\star} \left[     \rme^{  Z_1 X_1^\star +  (Z_2+r_m \sqrt{ 1+\frac{u^2}{r_p^2} }) X_2^\star   }  \right]   \right] +2\pi \nonumber \\
    &\overset{(d)}=   \log    \bbE_{\bfX^\star}  \left[      \rme^{  \frac{ (X_1^\star)^2}{2} +  \frac{ (X_2^\star)^2}{2} +r_m \sqrt{ 1+\frac{u^2}{r_p^2} } X_2^\star   }    \right] +2\pi \nonumber\\
    &\overset{(e)}\le  \log \rme^{  \frac{ r_p^2}{2} +  \frac{ r_m^2}{2} +r_m \sqrt{ 1+\frac{u^2}{r_p^2} } r_m   }     +2\pi \nonumber \\
    &=  \frac{ r_p^2}{2} +  \frac{ r_m^2}{2} +r_m^2 \sqrt{ 1+\frac{u^2}{r_p^2} }  +2\pi
\end{align}
where step $(a)$ follows by $| \log(z) | = | \log |z|  +i {\rm arg}(z) |  \le \log(|z|) +2\pi$; step $(b)$ follows from the triangle inequality;  step $(c)$ follows from using Jensen's inequality; step $(d)$ follows because the moment generation function of a standard Gaussian $Z$ is $\bbE[\rme^{tZ}]=\rme^{\frac{t^2}{2}}$; and step $(e)$ follows by $|X_1^\star| \le r_p$ and  $|X_2^\star| \le r_m$. 
\end{appendices}

\end{document}